\documentclass[lettersize,journal]{IEEEtran}
\usepackage{amsmath,amsfonts}
\usepackage{algorithmic}
\usepackage{algorithm}
\usepackage{array}
\usepackage[caption=false,font=normalsize,labelfont=sf,textfont=sf]{subfig}
\usepackage{textcomp}
\usepackage{stfloats}
\usepackage{url}
\usepackage{verbatim}
\usepackage{graphicx}
\usepackage{cite}
\usepackage{booktabs}
\usepackage{siunitx}
\usepackage{ragged2e}
\usepackage{tabularx}
\usepackage{adjustbox}

\usepackage{caption}
\usepackage{subcaption}

\usepackage{amsthm}

\newtheoremstyle{mytheoremstyle}
  {3pt}
  {3pt}
  {\itshape}
  {}
  {\bfseries}
  {.}
  {.5em}
  {}
  
\theoremstyle{mytheoremstyle}
\newtheorem{theorem}{\textbf{Theorem}}
\newtheorem{lemma}{\textbf{Lemma}}
\newtheorem{assumption}{\textbf{Assumption}}

\begin{document}

\title{pFedSOP : Accelerating Training Of Personalized Federated Learning Using Second-Order Optimization}
\author{Mrinmay Sen, Chalavadi Krishna Mohan,~\IEEEmembership{Senior Member,~IEEE}
\thanks{Mrinmay Sen and Chalavadi Krishna Mohan are with the Department of Artificial Intelligence, Indian Institute of Technology, Hyderabad 502285, India (e-mail: ai20resch11001@iith.ac.in; ckm@cse.iith.ac.in).}
\thanks{Manuscript received April 19, 2021; revised August 16, 2021.} }

\markboth{Journal of \LaTeX\ Class Files,~Vol.~14, No.~8, August~2021}%
{Shell \MakeLowercase{\textit{et al.}}: A Sample Article Using IEEEtran.cls for IEEE Journals}


\maketitle

\begin{abstract}
Personalized Federated Learning (PFL) enables clients to collaboratively train personalized models tailored to their individual objectives, addressing the challenge of model generalization in traditional Federated Learning (FL) due to high data heterogeneity. However, existing PFL methods often require increased communication rounds to achieve the desired performance, primarily due to slow training caused by the use of first-order optimization, which has linear convergence. Additionally, many of these methods increase local computation because of the additional data fed into the model during the search for personalized local models. One promising solution to this slow training is second-order optimization, known for its quadratic convergence. However, employing it in PFL is challenging due to the Hessian matrix and its inverse. In this paper, we propose pFedSOP, which efficiently utilizes second-order optimization in PFL to accelerate the training of personalized models and enhance performance with fewer communication rounds. Our approach first computes a personalized local gradient update using the Gompertz function-based normalized angle between local and global gradient updates, incorporating client-specific global information. We then use a regularized Fisher Information Matrix (FIM), computed from this personalized gradient update, as an approximation of the Hessian to update the personalized models. This FIM-based second-order optimization speeds up training with fewer communication rounds by tackling the challenges with exact Hessian and avoids additional data being fed into the model during the search for personalized local models. Extensive experiments on heterogeneously partitioned image classification datasets with partial client participation demonstrate that pFedSOP outperforms state-of-the-art FL and PFL algorithms.
\end{abstract}


\section{Introduction}

\IEEEPARstart{T}{he} generic federated learning (FL) \cite{McMahanaistats2017} is a distributed learning framework, where a group of data sources or clients collaboratively train a shared global model maintained by a central server, by performing infrequent aggregation of locally trained models in the server, enabling the local clients to collaborate with other clients without exchanging their locally stored raw data, thus helping to maintain local data privacy. A key challenge in this generic federated learning is the heterogeneous data across clients as different clients gather their personal data with different systems or processes, resulting in poor generalization performance of this single global model at local clients due to dissimilarity between local and global data distributions \cite{KairouzMABBBBCC21,FL_survey_HuangYSWLDY24}. In addition to data heterogeneity, local clients may have limited and scarce data, resulting in overfitting of the local model if we use conventional local training at the clients. To address these issues with generic federated learning and conventional local training with limited local data, personalized federated learning (PFL) \cite{Tancorr2021} has been proposed. It aims to collaboratively train a set of customized local models by using relevant global information in each of the local clients based on their local objectives. However, although existing PFL methods outperform traditional generic FL methods, they suffer from slow training due to the use of first-order optimization methods, such as stochastic gradient descent (SGD) \cite{amari1993backpropagation}, which exhibits linear convergence \cite{wright2006numerical}. SGD iteratively update model parameters ($\textbf{x}_\tau \in \mathbb{R}^d$) using only the first-order derivative of the objective function $P$ with respect to model parameters, i.e., gradient $\textbf{g}_\tau = \nabla P(\textbf{x}_{\tau-1}) \in \mathbb{R}^{d}$, as shown in Eq. \ref{eq1}. Here, the subscript $\tau$ refers to the values computed at iteration $\tau$.

\begin{equation}
\label{eq1}
    \textbf{x}_{\tau} = \textbf{x}_{\tau-1} - \eta \textbf{g}_\tau
\end{equation}
This leads to an increase in communication rounds in PFL while aiming to achieve the desired performance from the personalized local models. Additionally, these methods either result in higher local computation due to the additional data fed forward to the model during the search for personalized local models, or they suffer from communication inefficiencies by requiring the transfer of all client models from the server to each client. Specifically, the additional data forwarding during the personalization process increases local computation by $O(N_i d)$, and the total computation cost becomes $O(2N_i d)$, where $O(N_i d)$ accounts for the computation in the local training (similar to FedAvg), and an additional $O(N_i d)$ comes from the personalization process, here $N_i$ is number of local samples and $d$ is number of model parameters. One potential solution to the issue of slow training of personalized local models is to utilize second-order optimization (Newton method of optimization), known for its quadratic convergence rate \cite{Agarwaljmlr2017,Tankaria2021corr,Martens2015jmlr}, where the Hessian curvature along with the gradient is used while iteratively optimizing model parameters, as depicted in Eq. \ref{eq2}. Here, $\textbf{H}_\tau = \nabla^2 P(\textbf{x}_{\tau-1}) \in \mathbb{R}^{d \times d}$ represents the Hessian of the objective function computed at iteration $\tau$.
\begin{equation}
\label{eq2}
    \textbf{x}_{\tau} = \textbf{x}_{\tau-1} - \eta {\textbf{H}_\tau}^{-1} \textbf{g}_\tau
\end{equation}
Even though second-order optimization is advantageous in terms of a faster training, it has the drawback of requiring the computation and storage of the full Hessian matrix (a square matrix of size $d \times d$), which may not be suitable for large-scale models and large datasets. This raises questions about the practicality of applying second-order optimization, as it has computation cost of $O(N_id^2)$  for calculating the Hessian and a computation cost of $O(d^3)$ for inverting the Hessian. Additionally, it requires $O(d^2)$ space for storing the Hessian and its inverse, which can be a significant issue for resource-constrained devices. Another issue with utilizing second-order optimization in FL is that finding a global model by simply aggregating local models trained using the second-order Hessian is not feasible. This is because $\textbf{H}_t^{-1} \neq \frac{1}{K}\sum_{i = 1}^K \textbf{H}_{it}^{-1}$, where $\textbf{H}_{it}$ is the Hessian of the $i$-th client's local objective, $\textbf{H}_t$ is the average Hessian over all the clients (i.e., the global Hessian), and $K$ is number of clients. As a result, the aggregation process required in FL does not align with the requirements of second-order optimization, where the global model needs to be updated using the global Hessian, making its use difficult \cite{DDerezinskiM19,bischoff2021second}. 

To accelerate PFL training without the need for additional data input into the model, we propose pFedSOP, which aims to efficiently utilize second-order optimization while updating the personalized local models.This is achieved by addressing the computational and storage challenges associated with the Hessian and its inverse, using a regularized Fisher Information Matrix (FIM) \cite{sen2024sofim}, created through personalized local gradient updates, as a substitute for the Hessian in second-order optimization, since for probabilistic objective, the FIM is equivalent to the Hessian \cite{martens2020new}. In pFedSOP, each client computes the personalized local gradient update using the normalized angle between the local and global gradient updates, based on the Gompertz function \cite{gibbs2000variational} to capture relevant global information into the local gradient update. While employing FIM-based second-order optimization with this client-specific local gradient update, we apply the Sherman–Morrison formula for matrix inversion on the regularized FIM. This approach reduces the computational cost associated with the exact FIM and its inversion and eliminates the need to store the full FIM. The use of FIM with personalized local gradient update also removes the need to maintain a global model aggregated from local models trained with second-order optimization. Second-order optimization with the regularized FIM constructed using personalized local gradient updates in PFL enables faster training of personalized local models with fewer communication rounds, without the extra computation required for forwarding data to the model during the process of finding the personalized local models. Since pFedSOP only requires communication of the local gradient update from the client to the server and the global gradient update from the server to the client, it is efficient in terms of transmission costs, similar to FedAvg. The effectiveness of pFedSOP is validated through extensive experiments on several heterogeneously partitioned image classification datasets with partial client participation, where it is compared to state-of-the-art FL algorithms. The results demonstrate the superior performance of pFedSOP in terms of achieving better test accuracy and a faster reduction in training loss, with comparatively fewer communication rounds. Additionally, we provide the convergence analysis of the personalized local models trained using pFedSOP. As we use the FIM as the Hessian, our proposed pFedSOP is applicable to probabilistic objectives, such as image classification with the categorical cross-entropy loss function. The main contributions of this paper are as follows:
\begin{itemize} 
    \item We propose pFedSOP, a novel approach that reduces communication rounds in personalized federated learning (PFL) by accelerating the training of personalized local models through the efficient application of second-order optimization, leveraging its quadratic convergence.
    
    \item We propose a novel approach for constructing the regularized Fisher Information Matrix (FIM) using personalized local gradient update to approximate the Hessian in second-order optimization for PFL.
    
    \item  We directly derive the personalized local model using the Sherman–Morrison formula for matrix inversion on this regularized FIM, effectively addressing the computational difficulties associated with the exact Hessian and its inverse during personalized model updates.
    
    \item We introduce a new technique to compute personalized local gradient update based on Gompertz function that normalizes the angle between local and global gradients and gives the aggregation weight for them. This method effectively integrates client-specific global information into each client's gradient update.
    
    \item To the best of our knowledge, this is the first work to apply second-order optimization technique in PFL to accelerate the training of personalized models by addressing the challenges associated with the Hessian and its inverse.

\end{itemize}
The rest of the paper is organized as follows. section \ref{rel_w} shows related works, Section \ref{sc1} describes about preliminaries, section \ref{sc4} describes our proposed method, section \ref{sc5} shows our experiments $\&$ results and section \ref{sc6} concludes the paper.   

\section{Related work}
\label{rel_w}
The related works of improving FL training in heterogeneous data distribution can be  grouped into two categories. one is generic FL and another one is PFL. 
\subsection{Generic FL}
Exiting methods of generic FL includes FedProx \cite{Limlsys2020}, SCAFFOLD \cite{Karimireddyicml2020},  MOON \cite{Licvpr2021}, FedGA \cite{DandiBJgradalign2022}, FedCM \cite{xu2021fedcm}, FCCL+ \cite{HuangYSD24Similarity}, FedFA \cite{featureshiftZhouYWK24}, FedExP \cite{fedexp} etc.These methods can be viewed as the modifications on FedAvg \cite{McMahanaistats2017}. Same as FedAvg, this methods aims to optimize a single global model. To reduce the variability of the locally trained models caused by data heterogeneity, FedProx adds a proximal term ($\mu$) with the local objective. SCAFFOLD uses control variates in local training to handle the issue of drift between local and global updates caused by heterogeneous data across clients.  During local training, MOON utilizes contrastive learning while conducting local training to limit the variability between the local models and the global model. Instead of directly using the global model as an initial model during local training, FedGA modifies the global model with the measure of displacement of the global and local gradients before starting local training. FedExP uses a separate learning rate in the server, which is determined adaptively. To do so, FedExP utilizes extrapolation mechanism of Projection Onto Convex Sets (POCS) algorithm. FedCM also uses two separate learning rate like FedExP for both the server $\&$ clients and modifies the local gradient descent a weighted average of current local gradient and previous step's global gradient. FCCL+ tackles the challenges of data heterogeneity and catastrophic forgetting in Federated Learning (FL) by employing federated correlation and similarity learning, alongside non-target distillation. FedFA mitigates data heterogeneity arising from feature shifts caused by acquisition discrepancies across clients through federated feature augmentation.

Even these modifications on FedAvg perform better than FedAvg, when the label of data heterogeneity across clients are too high these methods fails to generalize well for all the clients as these methods are depended on a single global model \cite{Tancorr2021}. 

\subsection{Personalized FL}
To overcome these issue of generalization of the single global model in generic FL, researchers shifted their attention towards PFL, in which instead of training a single global model, each client is allowed to train its own customized model. Exiting PFL algorithms include FedBABU \cite{oh2021fedbabu}, FedRep \cite{collins2021exploiting}, FedPAC \cite{xu2023personalized}, FedCR \cite{zhang2023fedcr}, pFedMe \cite{DinhTN20NeurIPS}, Ditto \cite{li2021ditto},  FedPHP \cite{li2021fedphp}, FedAMP \cite{huang2021fedamp}, FedFomo \cite{zhang2020personalized}, FedDWA \cite{liu2023feddwa},  FedALA \cite{zhang2023fedala} etc. These methods can be grouped into three catagories as follows: \\
\textbf{Methods that divide the model :}
Taking inspiration from representation learning, FedRep, FedBABU and FedPAC divide each of the local model into feature extractor and  classifier head and only finds the global feature representation by aggregating all the locally trained feature extractors. Then fine tunes the classifier head in each client for making personalized model. FedCR also leverages a common representation learned across clients by employing a conditional mutual information (CMI) regularizer, which aligns local and global feature distributions.\\
\textbf{Methods that uses local regularization :}
pFedMe regularizes the  clients’ loss function using Moreau envelopes to personalize the local training. Similar to pFedMe, Ditto learns additional personalized model for each client using a proximal term with the local objective
to incorporate global information into local clients. \\
\textbf{Methods that performs personalized aggregation :}
FedPHP finds the personalized local model using moving
average of global and local models with a predefined hyperparameter. FedAMP uses attention-inducing function for making personalized model for each client. FedFomo finds personalized model for each client by aggregating other client's models with this client's model. The weights for aggregation is approximated using local models from other clients. FedFomo needs to download other clients models in each client.  Instead of downloading other clients' models, FedDWA performs personalized aggregation for each client on the server using its individual local guidance model, which is trained through local one-step adaptation. So, FedDWA requires to broadcast the local model and along with local guidance model from each client to the server. FedALA adaptively trains the aggregation weights of the global model and local models to find personalized models using local data and a local SGD optimizer, which requires additional training cost. 

Although existing PFL algorithms demonstrate superior generalization compared to generic FL methods, their training is often slow due to reliance on first-order optimization, which leads to increased communication rounds while aiming for specific performance from the personalized models. Additionally, these methods often increase local computation by forwarding data during the search for personalized local updates or suffer from communication inefficiencies due to sending all the received local models from the server to each client. In this paper, to address the issues with existing PFL methods, we propose pFedSOP, a novel PFL algorithm that efficiently utilizes second-order optimization with personalized local gradient update. This technique aligns with the last category of PFL methods.

\begin{figure*}[ht]
    \centering
    \begin{minipage}{0.7\textwidth}
        \centering
        \includegraphics[width=\linewidth]{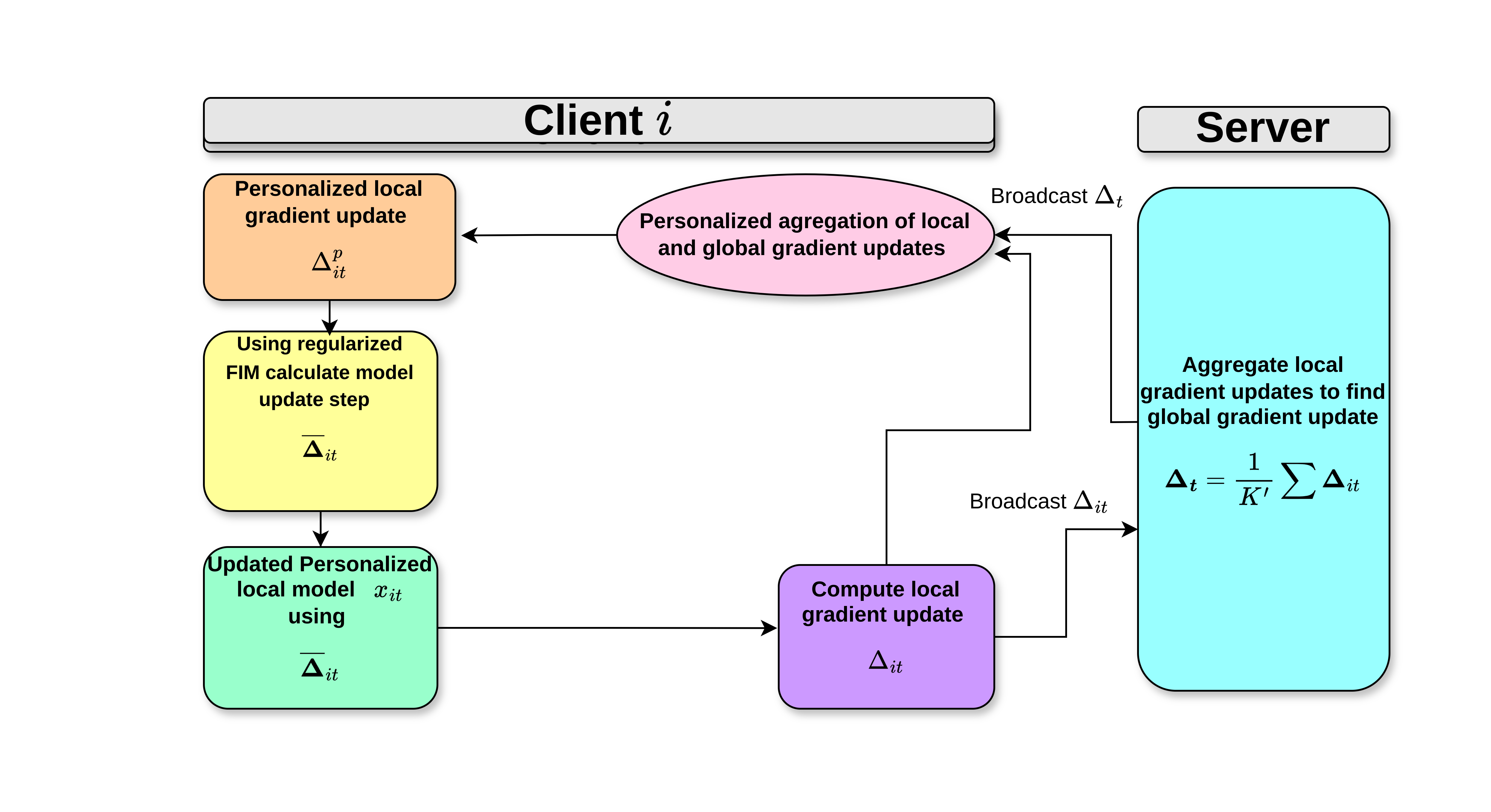}
    \end{minipage}
    \caption{Flow diagram of the proposed pFedSOP, illustrating the different operations performed at each client $i$ (personalized aggregation of local and global gradient updates, updating the personalized model, and computing the local gradient update), along with the server-side aggregation.}
    \label{fig:00flw}
\end{figure*}

\section{Preliminaries}
\label{sc1}
This section describes some preliminaries related to PFL and our designed algorithm.

\subsection{Problem Formulation of FL}

Generic FL is a distributed learning algorithm, which aims to collaboratively optimize a shared global model's parameters (\textbf{x}) by minimizing the average empirical risk objective $P(\textbf{x};\mathbb{D})$ computed across $K$ clients as shown in Eq. \ref{eq:3}.

\begin{equation}
\label{eq:3}
    \min_{\textbf{x}} P(\textbf{x};\mathbb{D})= \min_{\textbf{x}} \frac{1}{K}\sum^K_{i=1} P_i(\textbf{x};\mathbb{D}_i )
\end{equation}
where, $\mathbb{D}_i$ is the private dataset owned by client $i$, $\mathbb{D} = \bigcup_{i=1}^{K} \mathbb{D}_{i}$ , $P_i(\textbf{x};\mathbb{D}_i) = \frac{1}{|\mathbb{D}_i|}\sum_{\psi_j \in \mathbb{D}_i} f_j(\textbf{x}; \psi_j) $ is the empirical risk objective of client $i$ and $f_j(;)$ is the risk objective for the $j$-th sample of $\mathbb{D}_i$.

To solve the above problem described in Eq. \ref{eq:3}, the popular generic FL algorithm, called FedAvg \cite{McMahanaistats2017} allows individual clients to update the shared global model with their local dataset and local optimizer and the host server collects all the locally updated models $\{\textbf{x}_{i}\}^K_{i=1}$ and make a aggregation over these to find the updated global model as shown in Eq. \ref{eq:4} 

\begin{equation}
\label{eq:4}
    \textbf{x}= \frac{1}{K} \sum^K_{i=i} \textbf{x}_{i}.
\end{equation}

In personalized FL \cite{Tancorr2021,Fallahcorr2020,DinhTN20NeurIPS,zhang2023fedala,li2021ditto,zhang2023fedcr,liu2023feddwa,zhang2020personalized,xu2023personalized,collins2021exploiting}, the objective is to train a set of customized local models $\mathbb{X}= \{\textbf{x}_1, \textbf{x}_2,...,\textbf{x}_K\}$ by collaboratively optimize a set of empirical risk objectives $\{P_i(\textbf{x}_i;\mathbb{D}_i)\}^K_{i=1}$ computed across $K$ clients as shown in Eq. \ref{eq:5}

\begin{equation}
\label{eq:5}
    \min_{\mathbb{X}} P(\mathbb{X}; \mathbb{D})= \min_{\mathbb{X}} \frac{1}{K}\sum^K_{i=1} P_i(\textbf{x}_i;\mathbb{D}_i ).
\end{equation}
Optimization of the objective function mentioned in Eq. \ref{eq:5} leads to collaboratively optimize individual local model, which can perform well on its local data by taking relevant information from the global model while finding customized or personalized local models.

\subsection{Second-Order Optimization Using Fisher Information Matrix (FIM)}
The Fisher information matrix (FIM) $\textbf{F} \in \mathbb{R}^{d \times d}$ of an objective function $P(\textbf{x}; \mathbb{D})$ is defined as follows:
\begin{equation}
\label{eq:6}
    \textbf{F} = \mathbb{E}[\textbf{g}^i{\textbf{g}^i}^T] = \mathbb{E}_{\psi_i \in \mathbb{D}}[\nabla P(\textbf{x}; \psi_i) {\nabla P(\textbf{x}; \psi_i)}^T]
\end{equation}
where, $g^i = \nabla P(\textbf{x}; \psi_i)$ is the gradient of $P(\textbf{x}; \psi_i)$ at point \textbf{x} computed on the $i$-th sample $\psi_i \in \mathbb{D}$. It is proved that the FIM of negative log of a probabilistic objective (for example: Categorical cross entropy loss) is equivalent to the Hessian of that objective function \cite{martens2020new}. With this property, Natural Gradient Descent (NGD) \cite{zhang2019fast}, a variant of second-order optimization,  is proposed, which replaces the Hessian (\textbf{H}) with the FIM (\textbf{F}) in the update rule of second-order optimization (Eq. \ref{eq2}) for the probabilistic objective functions as shown in Eq. \ref{eq:7}. 

\begin{equation}
    \label{eq:7}
    \textbf{x}_{\tau} = \textbf{x}_{\tau-1} - \eta \textbf{F}_\tau^{-1} \textbf{g}_\tau
\end{equation}
Utilization of NGD with exact FIM is also a challenging task for large models and datasets, as it requires $O(Nd^2)$ computational cost for calculating FIM, similar to Eq. \ref{eq2}, which handles the exact Hessian. To effectively utilize NGD in PFL, we replace the true FIM with a regularized FIM, motivated by the paper of SOFIM \cite{sen2024sofim}. The regularized FIM is defined as follows :
\begin{equation}
    \label{eq:7a}
    \textbf{F} = \mathbb{E} [\textbf{g}^i {\textbf{g}^i}^T] \equiv \mathbb{E}[\textbf{g}^i] {\mathbb{E}[{\textbf{g}^i}]}^T + \rho \textbf{I}
\end{equation}
where, $\rho $ is regularization term, \textbf{I} $\in R^{d \times d}$ is Identity matrix and $\mathbb{E}[\textbf{g}^i]$ is the average gradient across all the samples $N$.

\subsection{Sherman Morrison Formula of Matrix Inversion}

The inverse of a matrix ${(\textbf{B} + \textbf{uv}^T)}^{-1}$ can be directly computed using Sherman Morrison formula of matrix inversion as shown in Eq. \ref{eq:7bb}.

\begin{equation}
    \label{eq:7bb}
    (\textbf{B} + \textbf{uv}^T)^{-1}= \textbf{B}^{-1} - \frac{\textbf{B}^{-1}\textbf{uv}^T \textbf{B}^{-1}}{1+\textbf{v}^T \textbf{B}^{-1}\textbf{u}}
\end{equation}
where, \textbf{B} $\in \mathbb{R}^{d \times d}$ is a invertible square matrix and \textbf{u}, \textbf{v} $\in \mathbb{R}^{d \times 1}$ are column vectors.

\section{Proposed method}
\label{sc4}
Our proposed pFedSOP is illustrated in Algorithm \ref{alg:algo3} and Figure \ref{fig:00flw}. In pFedSOP, at communication round $t$, each client $i$ first performs a personalized aggregation of the local gradient update $\boldsymbol{\Delta}_{i(t-1)}$ and the global gradient update $\boldsymbol{\Delta}_{(t-1)}$ from the previous communication round using the Gompertz function-based normalized angle between them. This is done to calculate the personalized local gradient update ${\boldsymbol{\Delta}^{p}_{it}}$, which captures relevant global information for each client. The client then uses a regularized FIM in second-order optimization with this personalized gradient update to update the personalized local model from $\textbf{x}_{i(t-1)}$ to $\textbf{x}_{it}$, as shown in Algorithm \ref{alg:algo1}. Since the new clients do not have a previous local gradient update, we randomly initialize their personalized models. Updating the personalized local model using the FIM constructed with the personalized local gradient update leads to faster training of the personalized models by efficiently handling the issues of the Hessian and its inverse in second-order optimization, while also avoiding the additional computation required for data feeding to the model during the personalization process. Once the personalized model is updated, the local gradient update $\boldsymbol{\Delta}_{it}$ at the point $\textbf{x}_{it}$ for this client $i$ is calculated using first-order optimization (SGD) on its local data, as shown in Algorithm \ref{alg:algo2}. The server then collects and aggregates the gradient updates from all clients to compute the global gradient update $\boldsymbol{\Delta}_{t}$ and broadcasts the global gradient update to all clients. To account for partial client participation in PFL (due to device label heterogeneity \cite{Limlsys2020}), pFedSOP randomly samples $K'$ clients from the total $K$ clients, with each client having an equal participation probability. During personalized aggregation in the partial client FL setting, pFedSOP considers the latest local gradient update from previous communication rounds as gradient update $\boldsymbol{\Delta}_{i(t-1)}$ for each client $i$. As our proposed pFedSOP uses FIM-based second-order optimization in PFL, this is limited to probabilistic objectives like the categorical cross-entropy loss function. All the steps associated with pFedSOP are summarized as follows:
\\
\textbf{step 1:} Each client $i$ performs a personalized aggregation of the local and global gradient updates and uses this aggregation to update its personalized local model through FIM-based second-order optimization.\\
\textbf{step 2:} Each client $i$ computes the local gradient update at the point of the updated personalized model using the SGD optimizer and local data.\\
\textbf{step 3:} The server collects and aggregates all the available local gradient updates, computes the global gradient update, and broadcasts it to all the available clients. \\

\begin{algorithm}[!h]
   \caption{Compute Personalized Model for Client $i$ at Communication Round $t$}
   \label{alg:algo1}
\begin{algorithmic}[1]
    \item \textbf{Input:} $\boldsymbol{\Delta}_{i(t-1)}$: Local gradient update, $\boldsymbol{\Delta}_{(t-1)}$: Global gradient update, $\textbf{x}_{i(t-1)}:$ Local model, $\eta_1$: Learning rate for personalization, $\rho$: Regularization parameter \newline
    \STATE Calculate cosine similarity between local and global gradient updates $(sim) =\frac{\boldsymbol{\Delta}_{i(t-1)} \cdot \boldsymbol{\Delta}_{(t-1)}}{||\boldsymbol{\Delta}_{i(t-1)}|| ||\boldsymbol{\Delta}_{(t-1)}||} \in [-1, 1]$
    \STATE Compute angle $\theta$ using inverse of cosine similarity ($\theta = arccosine(sim) \in [\pi, 0]$)
    \STATE Normalize this $\theta$ using Gompertz
function, $\beta = 1 - {e^{-e}}^{-\lambda(\theta - 1)} \in [0, 1]$, $\lambda > 0$
    \STATE Personalized aggregation of local and global gradient updates,  ${\boldsymbol{\Delta}^{p}_{it}} = (1-\beta) \boldsymbol{\Delta}_{i(t-1)} + \beta \boldsymbol{\Delta}_{(t-1)}$
    \STATE Calculate model update step $\overline{\boldsymbol{\Delta}}_{it}$  using Eq. \ref{eq:eq13} 
    \STATE Update local model $\textbf{x}_{it} = \textbf{x}_{i(t-1)} - \eta_1 \overline{\boldsymbol{\Delta}}_{it}$

\end{algorithmic}
\end{algorithm}

\begin{algorithm}[!h]
   \caption{Compute Local Gradient Update for Client $i$ at Communication Round $t$}
   \label{alg:algo2}
   
\begin{algorithmic}[1]
    \item \textbf{Input:} $\textbf{x}_{it}:$ Local personalized model, $\eta_2$: learning rate for local training, $\mathbb{D}_i$: Local dataset, $\mathcal{T}:$ Local iterations, $P_i:$ Local objective \newline
    \STATE $\textbf{x}^0_{it} =  
 \textbf{x}_{it}$
    \FOR{$\tau=1$ {\bfseries to} $\mathcal{T}$}
        \STATE Find stochastic gradient $ \textbf{g}^\tau_{it}= \frac{\partial P_i(\textbf{x}^{\tau-1}_{it}, \mathbb{D}_{i\tau})}{\partial \textbf{x}^{\tau-1}_{it}}$. where $\mathbb{D}_{i\tau} \subseteq \mathbb{D}_{i}$

        \STATE Update the model parameters $\textbf{x}^\tau_{it} \leftarrow \textbf{x}^{\tau-1}_{it} - \eta_2 \textbf{g}^\tau_{it}$
    \ENDFOR

    \STATE Find gradient update $\boldsymbol{\Delta}_{it} = \frac{\textbf{x}^0_{it} - \textbf{x}^{\mathcal{T}}_{it}}{\eta_2}$ 
\end{algorithmic}
\end{algorithm}

\subsection{Compute Local Gradient Update}
In pFedSOP, at each communication round  $t$, each client  $i$ first computes the personalized model ($\textbf{x}_{it}$) based on the similarity between the local and global gradient updates from the previous communication round. If the client is new, the personalized model for this client is randomly initialized. Once the personalized model $\textbf{x}_{it}$ is computed, client $i$ computes its local gradient update using Algorithm \ref{alg:algo2}, where it performs $\mathcal{T}$ iterations of SGD \cite{amari1993backpropagation} on the model initialized with the personalized model using its local data $\mathbb{D}_i$ , as shown in Eq. \ref{eq:eq8a}. The client then calculates the local gradient update $\boldsymbol{\Delta}_{it}$ using Eq. \ref{eq:eq8}, which is derived from the SGD update rule depicted in Eq. \ref{eq:eq9},

\begin{algorithm}[!h]
   \caption{pFedSOP}
   \label{alg:algo3}
\begin{algorithmic}[1]
    \item \textbf{Input:} $T$: Communication rounds, $\textbf{x}$: Randomly initialized model, $\eta_1$: learning rate for personalization, $\eta_2$: learning rate for local training, $\rho$: Regularization parameter,  \newline
    \item \textbf{Output:} Personalized local models $\{\textbf{x}_{it}\}^K_{i=1}$
    \FOR{$t=1$ {\bfseries to} $T$}
        \STATE Randomly pick a subset $\mathbb{C}$ of $K'$ clients from $K$ clients
         
        \STATE \underline{\textbf{In clients:}}\\
        \FOR{Client $i \in \mathbb{C}$   \textbf{in parallel} }
            \IF{ Client $i$ is new}
                \STATE  Initialize personalized model $\textbf{x}_{it} \leftarrow \textbf{x}$
            \ELSE 
                \STATE Compute personalized gradient update using $\boldsymbol{\Delta}_{i(t-1)}$ $\&$ $\boldsymbol{\Delta}_{(t-1)}$ and find updated personalized model $\textbf{x}_{it}$ as depicted in Algorithm \ref{alg:algo1}
            \ENDIF
            \STATE Calculate local gradient update $\boldsymbol{\Delta}_{it}$ at point $\textbf{x}_{it}$ using local data $\mathbb{D}_i$ and SGD optimizer as depicted in Algorithm \ref{alg:algo2}

        \ENDFOR
        \STATE \underline{\textbf{In server:}}\\
        \STATE Collect $\{\boldsymbol{\Delta}_{it}\}$ from all the available clients
        \STATE Aggregate all the local gradient updates to find the global gradient update $\boldsymbol{\Delta}_t = \frac{1}{K'}\sum \boldsymbol{\Delta}_{it}$ 
        
        \STATE Broadcast  $\boldsymbol{\Delta}_t$ to the clients available for next communication round
    \ENDFOR
\end{algorithmic}
\end{algorithm}

\begin{equation}
\label{eq:eq8a}
    \textbf{x}^\tau_{it} = \textbf{x}^{\tau-1}_{it} - \eta_2 \textbf{g}^\tau_{it}
\end{equation}

\begin{equation}
\label{eq:eq8}
    \boldsymbol{\Delta}_{it} = \frac{\textbf{x}^0_{it} - \textbf{x}^{\mathcal{T}}_{it}}{\eta_2}
\end{equation}
where, $\textbf{x}^0_{it} \leftarrow \textbf{x}_{it}$, $\textbf{g}^\tau_{it}= \frac{\partial P_i(\textbf{x}^{\tau-1}_{it}, \mathbb{D}_{i\tau})}{\partial \textbf{x}^{\tau-1}_{it}}$ is the stochastic gradient computed at $\tau$-th local iteration $(\mathbb{D}_{i\tau} \subseteq \mathbb{D}_{i})$, $\eta_2$ is the learning rate for local training. 

\begin{equation}
\label{eq:eq9}
    \textbf{x}^\mathcal{T}_{it} =\textbf{x}^0_{it} - \eta_2 \Delta_{it}
\end{equation}

The local gradient update calculated using Eq. \ref{eq:eq8} represents the sum of the gradient vectors across all the SGD iterations, meaning that $\Delta_{it} = \sum^{\mathcal{T}}_{\tau = 1} \textbf{g}^\tau_{it}$ is the vector representing the resultant gradient update of the set of gradient update vectors computed over $\mathcal{T}$ SGD iterations. After computing the local gradient update, client $i$ broadcasts it to the server.

\subsection{Calculate Global Gradient Update}
After receiving local gradient updates $\{ \Delta_{it}\}$ from all the available $K'$ clients, pFedSOP aggregates them  and finds the global gradient update $\Delta_t$ as shown in Eq. \ref{eq:eq10}. 

\begin{equation}
\label{eq:eq10}
    \Delta_t = \frac{1}{K'} \sum^{K'}_{j=1} \Delta_{jt}
\end{equation}

\subsection{Compute Personalized Local Gradient  Update}

At the beginning of each communication round $t > 1$, each client $i$ receives the global gradient update $\boldsymbol{\Delta}_{t-1}$ from the server and make a personalized aggregation of this global gradient with  its local gradient update $\boldsymbol{\Delta}_{i(t-1)}$ to calculate personalized local gradient update $\boldsymbol{\Delta}^p_{it}$ and uses this personalized gradient update to update the personalized local model from $\textbf{x}_{i(t-1)}$ to $\textbf{x}_{it}$ using FIM-based second order optimization as shown in Algorithm \ref{alg:algo1}. To do personalized aggregation, pFedSOP finds cosine similarity $ (sim) =\frac{\boldsymbol{\Delta}_{i(t-1)} \cdot \boldsymbol{\Delta}_{(t-1)}}{||\boldsymbol{\Delta}_{i(t-1)}|| ||\boldsymbol{\Delta}_{(t-1)}||} \in [-1, 1]$ between local and the global gradient updates of previous communication round and then calculate the angle $\theta  \in [\pi, 0]$ between them using inverse of the cosine similarity ( $\theta = {cos}^{-1} sim$). Once the angle between the local and global gradient updates is computed, it is then normalized with the Gompertz function \cite{gibbs2000variational} as shown in Eq. \ref{eq:eq1v1}.

\begin{equation}
\label{eq:eq1v1}
    \beta = 1 - {e^{-e}}^{-\lambda(\theta - 1)}
\end{equation}
where, $\beta \in [0, 1]$ is the normalized angle, which is used as an aggregation weight of $\boldsymbol{\Delta}_{i(t-1)}$ $\&$ $\boldsymbol{\Delta}_{(t-1)}$ and $\lambda > 0$ is a hyperparameter, which controls the steepness of the Gompertz function. After finding aggregation weight $\beta$, pFedSOP calculates the personalized local gradient update by performing the weighted average of $\boldsymbol{\Delta}_{i(t-1)}$ $\&$ $\boldsymbol{\Delta}_{(t-1)}$ as shown in Eq. \ref{eq:eq11}, which helps to adapt relavant information from the global gradient update into the local gradient update based on the similarity between them. In case of partial client participation in PFL, pFedSOP considers the latest local gradient update from previous communication rounds as gradient update $\boldsymbol{\Delta}_{i(t-1)}$ for each client $i$. 

\begin{equation}
    \label{eq:eq11}
    {\boldsymbol{\Delta}^p_{it}} = (1-\beta) \boldsymbol{\Delta}_{i(t-1)} + \beta \boldsymbol{\Delta}_{(t-1)}
\end{equation}

As $\boldsymbol{\Delta}_{it} = \sum^{\mathcal{T}}_{\tau = 1} \textbf{g}^\tau_{it}$ is the cumulative gradient over all the local iterations, we can write that,
\begin{equation}
    \label{eq:eq11a}
    \begin{aligned}
        {\boldsymbol{\Delta}^p_{it}} &= (1-\beta) \sum_{\tau = 1}^{\mathcal{T}} \textbf{g}^\tau_{it} + \beta \boldsymbol{\Delta}_{(t-1)} \\
        &=  \sum_{\tau = 1}^{\mathcal{T}} \left((1- \beta) \textbf{g}^\tau_{it} +  \frac{\beta \boldsymbol{\boldsymbol{\Delta}_{(t-1)}}}{\mathcal{T}} \right)\\
        &=  \sum_{\tau = 1}^{\mathcal{T}} {\textbf{g}_{it}^{p\tau}}
    \end{aligned}
\end{equation}
where, ${\textbf{g}_{it}^{p\tau}} = (1- \beta) \textbf{g}^\tau_{it} +  \frac{\beta \boldsymbol{\boldsymbol{\Delta}_{(t-1)}}}{\mathcal{T}}$ is the personalized gradient for $\tau$-th local iteration. With this, we can say that ${\boldsymbol{\Delta}^p_{it}}$ is the resultant personalized gradient over $\mathcal{T}$ local iterations for client $i$.

\subsection{Update the Personalized local Model}

Once the personalized local gradient update i.e. 
resultant personalized gradient ($\boldsymbol{\Delta}^p_{it}$) is calculated in client $i$, this is then used for updating the personalized local model $\textbf{x}_{i(t-1)}$ using Eq. \ref{eq:7}, where the FIM ($\textbf{F}_{it}$) is replaced by a regularized FIM to overcome the issue of high computation cost associated with exact FIM.  The regularized FIM with resultant personalized gradient is shown in Eq. \ref{eq:eq12}. 

\begin{equation}
    \label{eq:eq12}
    \textbf{F}_i \equiv [{\boldsymbol{\Delta}^p_{it}} {{\boldsymbol{\Delta}^p_{it}}}^T + \rho \textbf{I}]
\end{equation}

pFedSOP utilizes this regularized FIM as Hessian curvature and directly finds the update step ($\overline{\boldsymbol{\Delta}}_{it}$) for the personalized local model $\textbf{x}_{i(t-1)}$ using Sherman Morrison formula of matrix inversion by replacing $\textbf{u}=\textbf{v}={\boldsymbol{\Delta}^p_{it}}$ and $\textbf{B} = \rho \textbf{I}$ in Eq. \ref{eq:7bb} as shown below. 
\begin{equation}
    \label{eq:eq13}
    \begin{aligned}
        {\overline{\boldsymbol{\Delta}}_{it} } &= {\textbf{F}_{it}}^{-1} {\boldsymbol{\Delta}^p_{it}} \\
        &= {[{\boldsymbol{\Delta}^p_{it}} {{\boldsymbol{\Delta}^p_{it}}}^T + \rho \textbf{I}]}^{-1} {\boldsymbol{\Delta}^p_{it}} \\
        &=  \left[{(\rho \textbf{I})}^{-1} - \frac{{(\rho \textbf{I})}^{-1} {\boldsymbol{\Delta}^p_{it}} 
 {\boldsymbol{\Delta}^p_{it}}^T {(\rho \textbf{I})}^{-1}}{1+{\boldsymbol{\Delta}^p_{it}} ^T {(\rho \textbf{I})}^{-1}{\boldsymbol{\Delta}^p_{it}} } \right] {\boldsymbol{\Delta}^p_{it}} \\
        &=  \frac{{\boldsymbol{\Delta}^p_{it}}}{\rho} - \frac{ {\boldsymbol{\Delta}^p_{it}} 
 {\boldsymbol{\Delta}^p_{it}}^T {\boldsymbol{\Delta}^p_{it}} }{\rho^2+\rho {\boldsymbol{\Delta}^p_{it}}^T {\boldsymbol{\Delta}^p_{it}} }  
    \end{aligned}
\end{equation}

Once the update step ($\overline{\boldsymbol{\Delta}}_{it} $) is computed, this is then used for updating the local model $\textbf{x}_{i(t-1)}$ using Eq. \ref{eq:7} as shown below,

\begin{equation}
    \label{eq:eq14}
    \textbf{x}_{it} = \textbf{x}_{i(t-1)} - \eta_1 \overline{\boldsymbol{\Delta}}_{it}
\end{equation}
here, $\eta_1$ is the learning rate for updating the personalized local model. Using this regularized FIM, created with the personalized local gradient update, while updating the personalized model leads to faster training by efficiently utilizing second-order optimization without exposing high computation and storage costs. This, in turn, helps reduce the number of communication rounds required to achieve the desired level of performance from the personalized local models. Since pFedSOP uses FIM-based second-order optimization in PFL, it is limited to probabilistic local objectives, such as the categorical cross-entropy loss function for image classification.
\subsection{Convergence analysis}
This section examines the convergence of personalized local models trained using pFedSOP in PFL. The analysis demonstrates a linear-quadratic convergence guarantee for the local models under some proper choice of FIM regularization parameter $\rho$.

\begin{assumption}
\label{assump1}

The local objective function $P_i$ is  twice differentiable with respect to the local model parameters $\textbf{x}_{it}$.

\end{assumption}

\begin{assumption}
\label{assump2}

The Hessian $\textbf{H} = \nabla^2 P_i (\textbf{x}_{it})$ of the local objective function $P_i$ is L -Lipschitz continuous, that is \\
\[ ||\nabla^2 P_i (\textbf{x}_{it}) - \nabla^2 P_i (\textbf{x}_{i(t-1)})|| \leq L || \textbf{x}_{it} - \textbf{x}_{i(t-1)}||\].

\end{assumption}

\begin{lemma}
\label{lemma2}
$||\nabla^2 P_i (\textbf{x}^*_{i}) (\textbf{x}_{it} - \textbf{x}^*_{i})|| \leq \Gamma ||\textbf{x}_{it} - \textbf{x}^*_{i}||$, where $\Gamma$ is the largest eigenvalue of the Hessian  $\nabla^2 P_i(\textbf{x}^*_{i})$ at the optimum local model $\textbf{x}^*_{i}$.

\end{lemma}

\begin{proof}
If $\Gamma$ is the largest eigenvalue value of $\nabla^2 F(\textbf{x}^*_{i})$, we can write, $||\nabla^2 P_i (\textbf{x}^*_{i})|| = \Gamma$. Now, using the relationship between the matrix norm and vector norm, we can write that:

\[||\nabla^2 P_i (\textbf{x}^*_{i}) (\textbf{x}_{it} - \textbf{x}^*_{i})|| \leq ||\nabla^2 P_i (\textbf{x}^*_{i})|| . ||\textbf{x}_{it} - \textbf{x}^*_{i}||\]

Putting $\Gamma$ in place of $\nabla^2 P_i (\textbf{x}^*_{i})$, we get:

\[||\nabla^2 F(\textbf{x}^*_{i}) (\textbf{x}_{it} - \textbf{x}^*_{i})|| \leq \Gamma ||\textbf{x}_{it} - \textbf{x}^*_{i}||\]

\end{proof}

\begin{theorem}
\label{theorem__1}
If Assumptions \ref{assump1} and \ref{assump2} hold, then each local model $\textbf{x}_{it}$ can achieve a linear-quadratic convergence guarantee, given as follows :
\[ ||\textbf{x}_{it} - \textbf{x}^*_{i}|| \leq \varepsilon_1 ||\textbf{x}_{i(t-1)} - \textbf{x}^*_{i}|| + \varepsilon_2 {||\textbf{x}_{i(t-1)} - \textbf{x}^*_{i}||}^2 \]
\end{theorem}

Where, \[ {\varepsilon_1} = 1 + \frac{\Gamma \eta_1 }{\rho} + \frac{\Gamma \eta_1 {||\boldsymbol{\Delta}^p_{it} ||}^2}{\rho^2 + \rho {\boldsymbol{\Delta}^p_{it}}^T \boldsymbol{\Delta}^p_{it}}\]  \[{\varepsilon_2} = \frac{L \eta_1}{\rho} + \frac{L \eta_1 {||\boldsymbol{\Delta}^p_{it} ||}^2}{ \rho^2 + \rho {\boldsymbol{\Delta}^p_{it}}^T \boldsymbol{\Delta}^p_{it}}\].

\begin{proof}
From the update rule (Eq. \ref{eq:eq14}) of the personalized local model ($\textbf{x}_{it}$) of client $i$ in pFedSOP, we write that \\

\begin{equation}
\label{conv_eq:1}
    \textbf{x}_{it} = \textbf{x}_{i(t-1)} - \eta_1 \overline{\boldsymbol{\Delta}}_{it} 
\end{equation}

Where, ${\overline{\boldsymbol{\Delta}}_{it} } = {\textbf{F}_{it}}^{-1} {\boldsymbol{\Delta}^p_{it}}$ is the update step of the personalized local model, ${\textbf{F}_{it}} = [{\boldsymbol{\Delta}^p_{it}} {{\boldsymbol{\Delta}^p_{it}}}^T + \rho \textbf{I}]$ is the approximated FIM and ${\boldsymbol{\Delta}^p_{it}}$ is personalized local gradient. Using Sherman Morrison formula of matrix inversion, we can express ${\textbf{F}_{it}}^{-1}$ as follows:
\[{\textbf{F}_{it}}^{-1} = \frac{\textbf{I}}{\rho} - \frac{{\boldsymbol{\Delta}^p_{it}} {{\boldsymbol{\Delta}^p_{it}}}^T}{\rho^2 + \rho {{\boldsymbol{\Delta}^p_{it}}}^T {\boldsymbol{\Delta}^p_{it}}} \]
Suppose $\alpha_{1} = \frac{1}{\rho}$ and $\alpha_{2} = \frac{1}{\rho^2 + \rho {\boldsymbol{\Delta}^p_{it}}^T \boldsymbol{\Delta}^p_{it}}$. Using these in the above expression, we can write that

\[{\textbf{F}_{it}}^{-1} = \alpha_{1} \textbf{I} - \alpha_{2} \boldsymbol{\Delta}^p_{it} {\boldsymbol{\Delta}^p_{it}}^T. \]
Using the above expression of ${\textbf{F}_{it}}^{-1}$  and subtracting $\textbf{x}^*_{i}$ from both side in Eq. \ref{conv_eq:1}, we get:

\begin{equation}
\label{conv_eq:2}
    \textbf{x}_{it} - \textbf{x}^*_{i} = \textbf{x}_{i(t-1)} -\textbf{x}^*_{i} - \eta_1 (\alpha_{1} \textbf{I} - \alpha_{2} \boldsymbol{\Delta}^p_{it} {\boldsymbol{\Delta}^p_{it}}^T) \boldsymbol{\Delta}^p_{it} 
\end{equation}
Suppose $\textbf{e}_t = \textbf{x}_{it} -\textbf{x}^*_{i}$ and $\textbf{e}_{t-1} = \textbf{x}_{i(t-1)} -\textbf{x}^*_{i}$ are the errors at $t$ - th and $(t-1)$ - th communication rounds for client $i$. Putting these terms in  Eq. \ref{conv_eq:2}, we get:

\begin{equation}
\label{conv_eq:3}
    \textbf{e}_t = \textbf{e}_{t-1} - \eta_1 (\alpha_{1} \textbf{I} - \alpha_{2} \boldsymbol{\Delta}^p_{it} {\boldsymbol{\Delta}^p_{it}}^T) \boldsymbol{\Delta}^p_{it} 
\end{equation}
From the update rule of second-order optimization with exact Hessian, we can write that: 

\begin{equation}
\label{conv_eq:3prev}
    \textbf{x}^*_{i} = \textbf{x}_{i(t-1)} - {(\nabla^2 P_i(\textbf{x}_{i(t-1)}))}^{-1} \boldsymbol{\Delta}^p_{it} 
\end{equation}
here, $\nabla^2 P_i(\textbf{x}_{i(t-1)})$ is the true Hessian of the local objective calculated at $\textbf{x}_{i(t-1)}$.
From the Eq. \ref{conv_eq:3prev}, we get, 
\[\textbf{x}_{i(t-1)} - \textbf{x}^*_{i}  = {(\nabla^2 P_i(\textbf{x}_{i(t-1)}))}^{-1} \boldsymbol{\Delta}^p_{it}  \]
\[\boldsymbol{\Delta}^p_{it} = (\nabla^2 P_i(\textbf{x}_{i(t-1)})) \textbf{e}_{t-1}\]
Replacing the above expression of $\boldsymbol{\Delta}^p_{it}$ in Eq. \ref{conv_eq:3}, we get:
\[\textbf{e}_t = \textbf{e}_{t-1} - \eta_1 (\alpha_{1} \textbf{I} - \alpha_{2} \boldsymbol{\Delta}^p_{it} {\boldsymbol{\Delta}^p_{it}}^T) (\nabla^2 P_i(\textbf{x}_{i(t-1)})) \textbf{e}_{t-1}  \]
Let’s say $\textbf{G} = \eta_1 (\alpha_{1} \textbf{I} - \alpha_{2} \boldsymbol{\Delta}^p_{it} {\boldsymbol{\Delta}^p_{it}}^T)$. Then, we can write:
\[\textbf{e}_t = \textbf{e}_{t-1} - \textbf{G} \left(\nabla^2 P_i(\textbf{x}_{i(t-1)}) - \nabla^2 P_i (\textbf{x}^*_{i}) + \nabla^2 P_i (\textbf{x}^*_{i})\right)  \textbf{e}_{t-1}
    \].
\begin{equation}
\label{conv_eq:4}
\begin{gathered}
    \textbf{e}_t = \textbf{e}_{t-1} - \textbf{G} \left(\nabla^2 P_i (\textbf{x}_{i(t-1)}) - \nabla^2 P_i(\textbf{x}^*_{i})\right)  \textbf{e}_{t-1}\\ 
    - \textbf{G} \nabla^2 P_i(\textbf{x}^{*}_{i}) \textbf{e}_{t-1} 
\end{gathered}
\end{equation}
Taking norm both sides of Eq. \ref{conv_eq:4},

\begin{equation}
\label{conv_eq:5}
\begin{gathered}
    ||\textbf{e}_t|| = ||\textbf{e}_{t-1} - \textbf{G} \left(\nabla^2 P_i (\textbf{x}_{i(t-1)}) - \nabla^2 P_i(\textbf{x}^*_{i})\right)  \textbf{e}_{t-1}\\ 
    - \textbf{G} \nabla^2 P_i(\textbf{x}^{*}_{i}) \textbf{e}_{t-1}|| 
\end{gathered}
\end{equation}
 
\begin{equation}
\label{conv_eq:6}
\begin{gathered}
    ||\textbf{e}_t|| \leq ||\textbf{e}_{t-1}|| + ||\textbf{G} \left(\nabla^2 P_i (\textbf{x}_{i(t-1)}) - \nabla^2 P_i(\textbf{x}^*_{i})\right)  \textbf{e}_{t-1}\\ 
    + \textbf{G} \nabla^2 P_i(\textbf{x}^{*}_{i}) \textbf{e}_{t-1}|| 
\end{gathered}
\end{equation}
If assumption \ref{assump2} holds, we can write that:
\begin{equation}
\label{conv_eq:7}
\begin{gathered}
    ||\textbf{e}_t|| \leq ||\textbf{e}_{t-1}|| + ||\textbf{G} L  {\textbf{e}_{t-1}} {\textbf{e}_{t-1}}^T||  + ||\textbf{G} \nabla^2 P_i(\textbf{x}^{*}_{i}) \textbf{e}_{t-1}|| 
\end{gathered}
\end{equation}
Utilizing Lemma \ref{lemma2}, we can rewrite Eq. \ref{conv_eq:7} as shown below :

\[||\textbf{e}_t|| \leq ||\textbf{e}_{t-1}|| + ||\textbf{G} L  {\textbf{e}_{t-1}} {\textbf{e}_{t-1}}^T|| + ||\textbf{G} \Gamma \textbf{e}_{t-1}|| \]

\begin{equation}
\label{conv_eq:8}
\begin{gathered}
    ||\textbf{e}_t|| \leq \underline{\left(1 + \Gamma ||\textbf{G}||\right)} ||\textbf{e}_{t-1}|| + \underline{\left(L ||\textbf{G}||\right)} {||\textbf{e}_{t-1}||}^2 
\end{gathered}
\end{equation}

From $||\textbf{G}||$, we find that:
\begin{equation}
\label{conv_eq:9}
  \begin{aligned}
    ||\textbf{G}|| & =||\eta_1 (\alpha_{1} \textbf{I} - \alpha_{2} \boldsymbol{\Delta}^p_{it} {\boldsymbol{\Delta}^p_{it}}^T)||\\
      & \leq ||\eta_1 \alpha_{1} \textbf{I}|| + ||\eta_1 \alpha_{2} \boldsymbol{\Delta}^p_{it} {\boldsymbol{\Delta}^p_{it}}^T||\\
      & \leq \eta_1 \alpha_{1} + \eta_1 \alpha_2  {||\boldsymbol{\Delta}^p_{it}||}^2
  \end{aligned}
\end{equation}
Using Eq. \ref{conv_eq:9} and Eq. \ref{conv_eq:8}, we can write that:

\begin{equation}
\label{conv_eq:10}
\begin{gathered}
    ||\textbf{e}_t|| \leq (1 + \Gamma \eta_1 \alpha_{1} + \Gamma \eta_1 \alpha_{2} {||\boldsymbol{\Delta}^p_{it}||}^2) {||\textbf{e}_{t-1}||}\\
    + (L \eta_1 \alpha_{1} + L \eta_1 \alpha_{2} {||\boldsymbol{\Delta}^p_{it}||}^2 )|| {\textbf{e}_{t-1}||}^2 
\end{gathered}
\end{equation}
Replacing expressions of $\alpha_{1}$, $\alpha_{2}$, $\textbf{e}_t$ and $\textbf{e}_{t-1}$ in Eq. \ref{conv_eq:10}, we find: 

\[ ||\textbf{x}_{it} - \textbf{x}^{*}_{i}|| \leq \varepsilon_1 ||\textbf{x}_{i(t-1)} - \textbf{x}^{*}_{i}|| + \varepsilon_2 {||\textbf{x}_{i(t-1)} - \textbf{x}^{*}_{i}||}^2 \]
Where, 

\[ {\varepsilon_1} = 1 + \frac{\Gamma \eta_1 }{\rho} + \frac{\Gamma \eta_1 {||\boldsymbol{\Delta}^p_{it} ||}^2}{\rho^2 + \rho {\boldsymbol{\Delta}^p_{it}}^T \boldsymbol{\Delta}^p_{it}}\]  \[{\varepsilon_2} = \frac{L \eta_1}{\rho} + \frac{L \eta_1 {||\boldsymbol{\Delta}^p_{it} ||}^2}{ \rho^2 + \rho {\boldsymbol{\Delta}^p_{it}}^T \boldsymbol{\Delta}^p_{it}}\]
\end{proof}
\textbf{Analysis : } When \(\eta_1\) and \(\Gamma\) are fixed, \(\varepsilon_1\) decreases as \(\rho\) increases, leading to faster linear convergence \cite{wright2006numerical}. However, if \(\rho\) becomes too large, it can result in numerical instability. Similarly, \(\varepsilon_2\) decreases as \(\rho\) increases (when  \(\eta_1\) and \(L\) are fixed), indicating faster quadratic convergence \cite{wright2006numerical}, but excessively large values of \(\rho\) can also cause instability. Values of \(\rho\) that are either too large or too small can result in instability when computing the inverse Hessian. Hence, selecting an appropriate range for \(\rho\) is crucial to ensure numerical stability while optimizing the convergence rate.

\subsection{Computation and transmission costs}
\begin{table*}[ht]
\centering
\newcolumntype{M}{>{\centering\arraybackslash}m{1.3cm}} 
\newcolumntype{P}{>{\raggedright\arraybackslash}p{4cm}} 

\begin{tabular}{|M|M|M|M|M|M|M|M|M|M|}
\hline
\textbf{Methods} & \textbf{FedAvg} & \textbf{FedProx} & \textbf{FedAvg FT} & \textbf{FedProx FT} & \textbf{Ditto} & \textbf{FedRep} & \textbf{FedALA} & \textbf{FedDWA} & \textbf{pFedSOP} \\
\hline
Local computation   & $O(N_i d)$ & $O(N_i d)$ & $O(N_i d + N_i d)$ & $O(N_i d + N_i d)$ & $O(N_i d + N_i d)$ & $O(N_i d)$ & $O(N_i d + N_i d)$ & $O(N_i d + N_i d)$ & $O(N_i d + 2d)$ \\
\hline
Server computation & $O(K'd)$ & $O(K'd)$ & $O(K'd)$ & $O(K'd)$ & $O(K'd)$ & $O(K' d_r)$ & $O(K'd)$ & $O({K'}^2 d)$ & $O(K'd)$ \\
\hline
\end{tabular}
\caption{Comparisons of local and server computation costs among different FL methods. Here, $N_i$ represents number of samples in $i$ - th client, $d$ is number of model parameters, $K'$ is number of available clients in each FL iteration and $d_r$ is the number of parameters in the feature extractor part of the model.}
\label{tab:complx}
\end{table*}
The computation costs associated with different methods are depicted in Table \ref{tab:complx}. From this table, it can be noticed that due to additional data feeding to the local model while finding personalized local update, existing PFL methods such as FedAvg FT, FedProx FT, Ditto, FedALA, FedDWA etc requires additional $O(N_i d)$ local computation, where in pFedSOP requires additional $O(2d)$ local computation as compared to generic FL i.e. FedAvg. Here, $N_i$ is the number of samples in $i$ - th client. So, we can say that the local computation cost of pFedSOP is competitive to FedAvg and less than previously mentioned existing PFL methods. The server computation cost of pFedSOP is similar to most FL methods but lower than that of FedDWA. As pFedSOP is associated with transmission of gradient updates in both the directions (server to client and client to server), the transmission cost of pFedSOP is $O(2d)$, which is the same as FedAvg, whereas FedDWA incurs a transmission cost of $O(2d)$, plus an additional $O(d)$ costs from client to the server.
\section{Experiments}
\label{sc5}
The proposed PFL algorithm, pFedSOP, is validated through extensive experiments on several heterogeneously partitioned image-classification datasets. It is compared with state-of-the-art PFL algorithms, as well as FedAvg and FedProx. This section details the datasets used, the model employed, the methods compared, evaluation metrics, implementation details, results, and analysis. The experimental results demonstrate the effectiveness of pFedSOp in achieving better test accuracy and a faster decrease in training loss while optimizing customized local models. Additionally, we analyze the effect of our designed personalization component (PC) in the PFL algorithm and conduct a sensitivity analysis of the hyperparameters. 
\subsection{Datasets}
We conduct extensive experiments on three benchmark image classification datasets, namely, CIFAR10, CIFAR100 \cite{krizhevsky2009learning} and Tiny ImageNet \cite{chrabaszcz2017downsampled}. The CIFAR-10 dataset consists of 50,000 training samples and 10,000 test samples. Each sample is a color image of size 3 x 32 x 32. CIFAR-10 includes 10 classes, with 6,000 samples per class. Similarly, the CIFAR-100 dataset also contains color images of size 3 x 32 x 32, with 50,000 training samples and 10,000 test samples. However, CIFAR-100 has 100 classes, with 600 samples per class. Tiny ImageNet features 100,000 training samples and 10,000 test samples across 200 classes. Each class has overall 550 samples, and each sample is a color image of size 3 x 64 x 64. In our experiments, we merge train and test samples for each of the datasets and create heterogeneous partitions across $K = 100$ number of clients.

For each dataset, we conduct experiments under two different heterogeneous data settings. The first setting is inspired by the FedDWA paper \cite{liu2023feddwa}, where the dataset is divided across clients using a Dirichlet distribution with a Dirichlet parameter of $\text{Dir} = 0.07$. The second setting, referred to as the pathological heterogeneous setting, is inspired by the FedALA paper \cite{zhang2023fedala}. In this setting, we create $s$ number of shards, each of size $z$, such that $s \times z$ equals the total number of samples. We then divide $b$ number of shards to each of the $K = 100$ clients. Following FedALA, we use $ z = $200, 600 and 1000 for CIFAR-10, CIFAR-100, and Tiny ImageNet, respectively. This means that each client receives samples from $b = 2$ classes for CIFAR10, $b = 6$ classes for CIFAR100 and $b = 10$ classes for Tiny ImageNet. 

After creating heterogeneous partitions across clients, each client's data is divided into a training set and a test set. We randomly assign 80 $\%$ of the total samples for training, and the remaining 20 $\%$ is kept for evaluation.

\subsection{Experimental Setup}

In this section, we discuss about models used for classification tasks, Compared methods, performance metrics and implementation details. 

\subsubsection{Models and loss function}

For federated image classification of CIFAR10 dataset, we use Resnet 18 \cite{he2016deepresnet} convolutional neural Network (CNN) model. For CIFAR100 and Tiny ImageNet, we use Resnet 9 \cite{he2016deepresnet} CNN model. For training the local model, we use categorical cross entropy loss function.
\subsubsection{Compared methods}
We compare our proposed method, pFedSOP, with generic FL methods such as FedAvg and FedProx, as well as state-of-the-art PFL methods such as Ditto, FedRep, FedALA, and FedDWA. Details about these methods are mentioned in the related work section (Section \ref{rel_w}). Additionally, we compare our algorithm with FedAvg FT and FedProx FT, where the global model (of FedAvg or FedProx) is fine-tuned using local data before being used as the initial model for local training.

\subsubsection{Performance metrics}
In each communication round, we calculate the test accuracy (=  $\frac{\text{Number of true predictions}}{\text{Number of samples}}$) of the personalized local model for each participating client, based on their own local test data. We then compute the average test accuracy across all participating clients. Additionally, we measure the average training loss across all participating clients in each communication round. We also track the highest accuracy achieved by each client within the predefined communication round and compute the average of these highest accuracies across all clients.  
\begin{table*}[h!]
\centering
\caption{Best accuracy achieved and the average time (seconds) taken in each communication round by different methods on the CIFAR-10, CIFAR-100, and Tiny ImageNet datasets are reported for both heterogeneous settings (Dir(0.07) and Pathological) within the predefined communication rounds $T = 100$.}
\newcolumntype{l}{>{\centering\arraybackslash}m{1.2cm}}
\begin{tabular}{*{11}{l}}
\hline
\hline
\multicolumn{2}{c}{} & \multicolumn{3}{c}{CIFAR10} & \multicolumn{3}{c}{CIFAR100} & \multicolumn{3}{c}{Tiny ImageNet} \\
\cmidrule(r){3-5}  \cmidrule(lr){6-8} \cmidrule(lr){9-11} 

\multicolumn{2}{c}{Method} & Dir(0.07) & Pathological & Average time/round & Dir(0.07) & Pathological & Average time/round & Dir(0.07) & Pathological & Average time/round  \\
\hline
\hline
\multicolumn{2}{c}{FedAvg}  & 73.12 & 71.97 & 29.85 & 34.79 & 33.18 & 28.86 & 27.31 & 27.82 & 60.71  \\
\multicolumn{2}{c}{FedProx}  & 72.84 & 72.79 & 30.52 & 34.21 & 32.77 & 30.14 &26.54&27.03&63.69  \\
\multicolumn{2}{c}{FedAvg FT}  & 86.85 & 88.83& 43.54 & 55.65 & 65.81& 47.68&44.61&57.54&106.81 \\
\multicolumn{2}{c}{FedProx FT}  & 86.20 & 89.31& 45.33 & 55.30 & 65.19& 48.99&43.79&56.86&109.11  \\
\multicolumn{2}{c}{Ditto}  & 83.60 & 89.82& 47.29 & 57.99 & 71.47& 51.13&42.65&57.76&111.58  \\
\multicolumn{2}{c}{FedRep}  & 85.36 & 86.76& 31.99 & 56.91 & 68.52& 35.37&38.05&44.26&76.73  \\
\multicolumn{2}{c}{FedALA}  & 86.65 & 89.44& 48.40 & 57.69 & 67.92& 49.42&44.89&57.48&109.96 \\
\multicolumn{2}{c}{FedDWA}  & 88.64 & 91.80& 47.09 & 66.75 & 78.18& 49.93&60.18&69.42&110.51 \\
\multicolumn{2}{c}{pFedSOP}  & \textbf{90.64} & \textbf{93.63} & 30.55 & \textbf{71.36} & \textbf{82.75}& 29.38&\textbf{65.42}&\textbf{74.20}&60.09  \\
\hline
\hline
\end{tabular}
\label{tab:1}
\end{table*}
\subsubsection{Implementation details}
We conduct our experiments with partial client participation instead of full client participation in each communication round. Similar to FedDWA, we consider 20 $\%$ client participation out of $K = 100$ clients. This selection of partial clients in each communication round is executed in such a way that each client is selected with equal probability to participate in the FL training. To find the best performance from each method, we experiment with different sets of hyperparameters. The best set of hyperparameters for each method is determined by considering the minimum training loss and the maximum test accuracy across all available clients and communication rounds. We compare our proposed method with others based on this best performance. We use a learning rate $\in \{1, 0.1, 0.01, 0.001, 0.0001\}$ for all methods, FedProx's proximal term $\mu \in \{1, 0.1\}$, Ditto's $\lambda \in \{1, 0.1, 0.01\}$, and set the regularization term of pFedSOP to $\rho = 1$. The $\lambda$ for pFedSOP is also set to 1, and the local batch size is 50. For all methods, we use one local epoch. All experiments are conducted using a Tesla V100 GPU and PyTorch 1.12.1+cu102. We use seed = 0 to reproduce the results. While conducting our experiments, we use same initialization and same settings for all the methods for each heterogeneous FL setup of each dataset, which ensure the fairness. 

\begin{figure}[ht]
    \centering
    \begin{minipage}{0.4\textwidth}
        \centering
        \includegraphics[width=\linewidth]{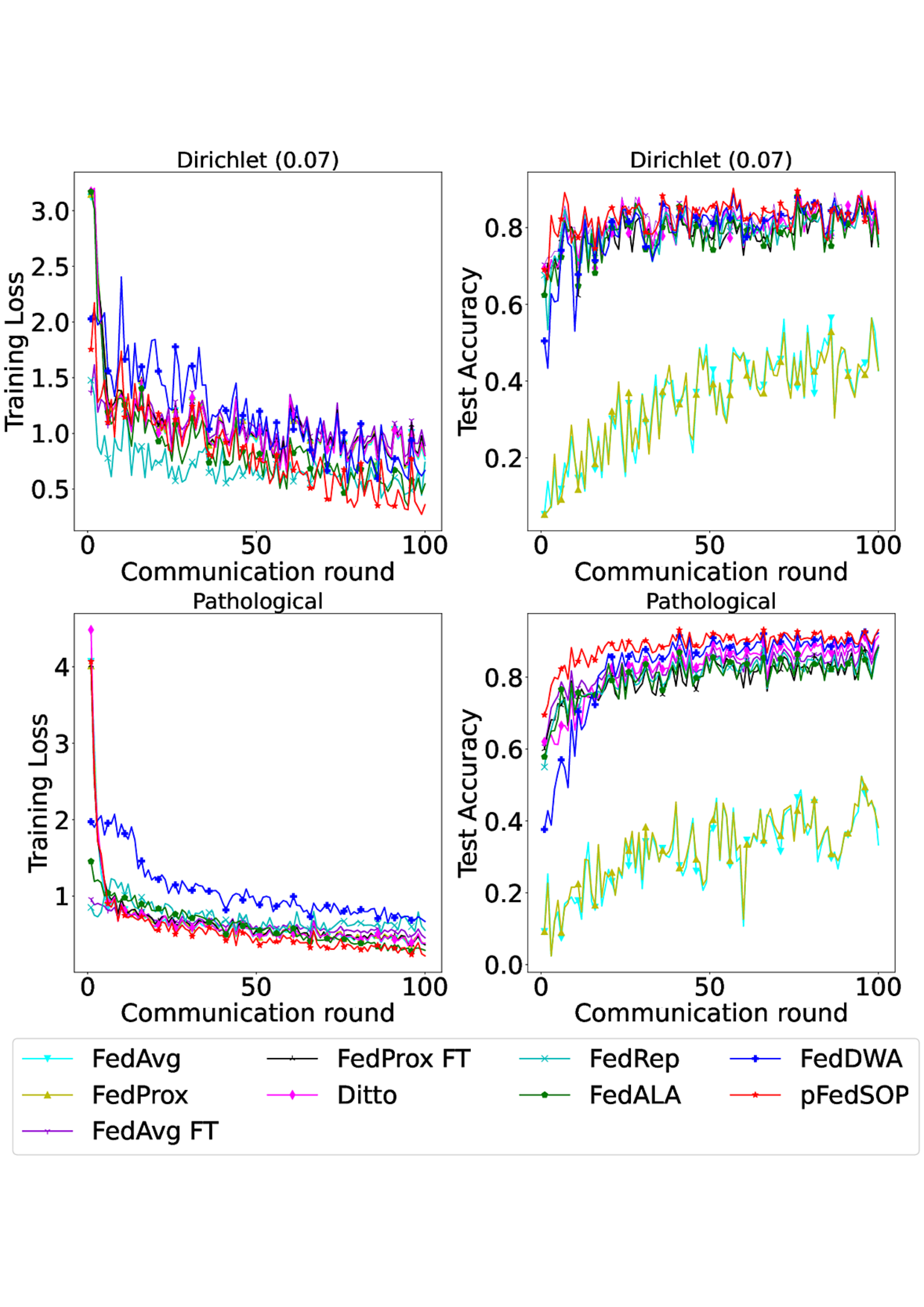}
    \end{minipage}\hfill
    \caption{Comparison of various methods in terms of communication round-wise average training loss, as well as average test accuracy, across all available clients for both heterogeneous settings of the CIFAR10 dataset. The top row shows results from the Dirichlet-based federated learning setting, while the bottom row displays results from the pathological federated learning setting.}
    \label{fig:1}
\end{figure}

\begin{figure}[ht]
    \centering
    \begin{minipage}{0.4\textwidth}
        \centering
        \includegraphics[width=\linewidth]{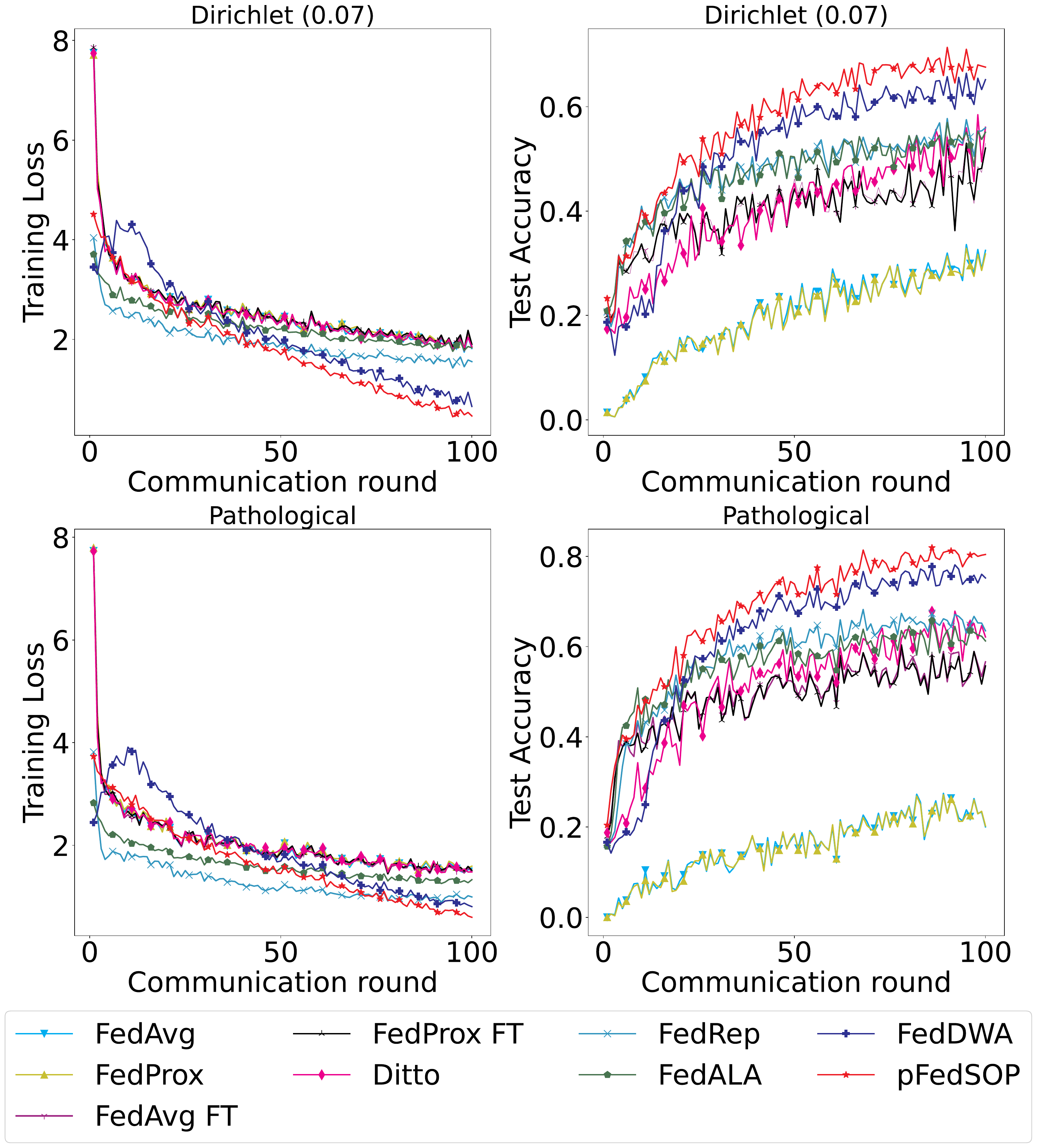}
    \end{minipage}\hfill
    \caption{Comparison of various methods in terms of communication round-wise average training loss, as well as average test accuracy, across all available clients for both heterogeneous settings of the CIFAR100 dataset. The top row shows results from the Dirichlet-based federated learning setting, while the bottom row displays results from the pathological federated learning setting.}
    \label{fig:2}
\end{figure}
\begin{figure}[ht]
    \centering
    \begin{minipage}{0.4\textwidth}
        \centering
        \includegraphics[width=\linewidth]{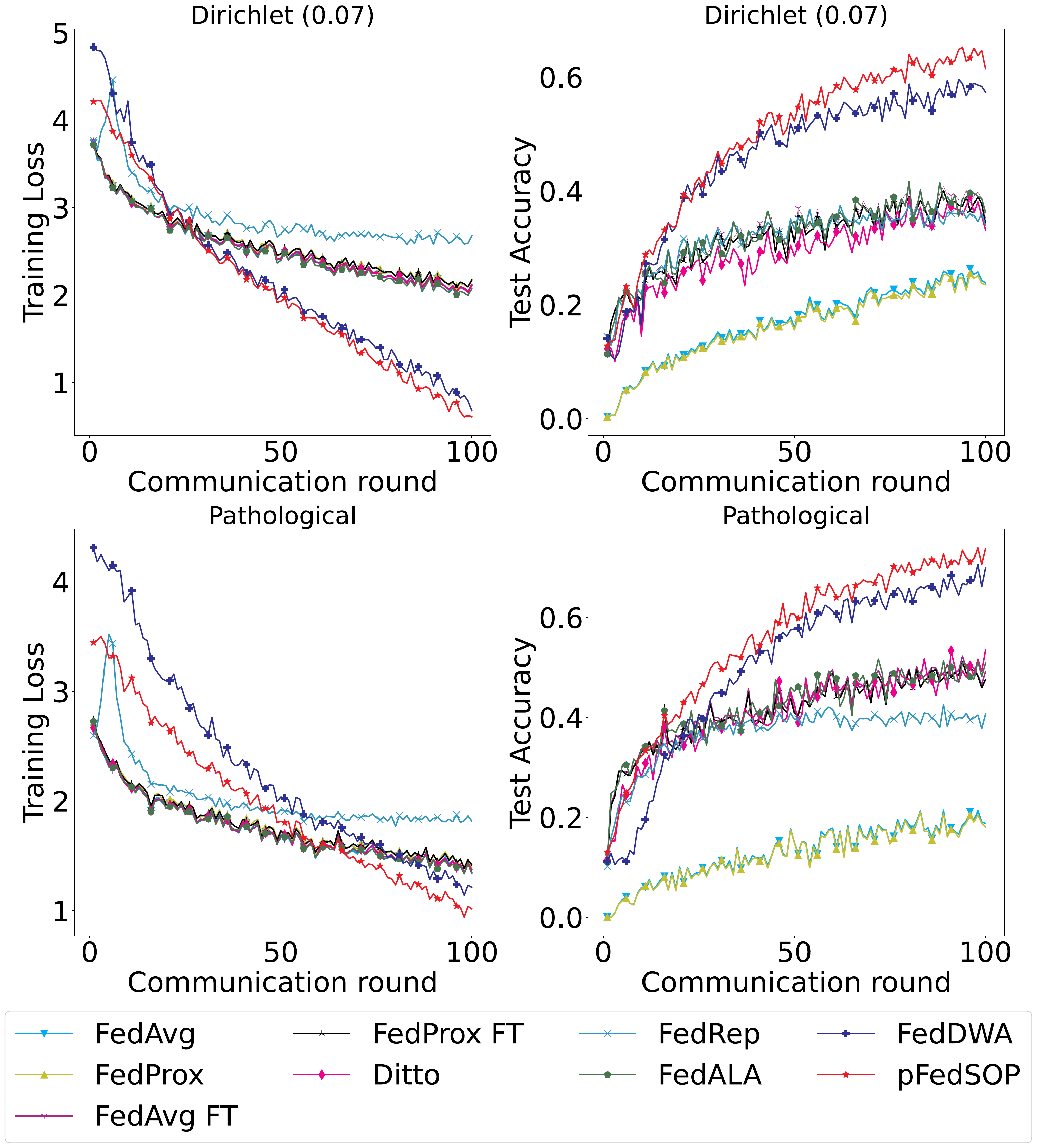}
    \end{minipage}\hfill
    \caption{Comparison of various methods in terms of communication round-wise average training loss, as well as average test accuracy, across all available clients for both heterogeneous settings of the Tiny ImageNet dataset. The top row shows results from the Dirichlet-based federated learning setting, while the bottom row displays results from the pathological federated learning setting.}
    \label{fig:3}
\end{figure}

\subsection{Results}
This section demonstrates the results of our experiments in both heterogeneous settings and compares the proposed method with existing FL methods in terms of efficiency, focusing on reducing training loss and achieving better test accuracy. Additionally, this section analyzes the effectiveness of our designed personalization component and examines the sensitivity of the hyperparameters.

\subsubsection{Comparisons with existing FL methods}
To validate the effectiveness of pFedSOP, we compare it with existing methods in terms of lower losses and higher test accuracy across all clients. This comparison is based on extensive experiments conducted on several benchmark datasets, with two different heterogeneous FL settings for each dataset. The experimental results are presented in Table \ref{tab:1} and Figures \ref{fig:1}, \ref{fig:2} $\&$ \ref{fig:3}. Table \ref{tab:1} shows the highest test accuracy achieved by different methods across all clients over all the communication rounds. To construct this table, we record the highest accuracy achieved by each client across all communication rounds and calculate the average of these highest accuracies for each method. This is done separately for both heterogeneous settings and each dataset. Additionally, we include the average time taken by each method to complete one communication round for each dataset. From this table, it can be observed that pFedSOP achieves the highest test accuracy within the predefined communication rounds $T=100$ compared to existing methods in both heterogeneous settings for each dataset. From this table, it can also be observed that pFedSOP requires a competitive amount of time to complete one communication round, similar to FedAvg, FedProx, and FedRep, but less than the time required by FedAvg FT, FedProx FT, Ditto, FedALA, and FedDWA. This indicates that pFedSOP has lower computational complexity compared to FedAvg FT, FedProx FT, Ditto, FedALA, and FedDWA, and competitive complexity compared to FedAvg, FedProx, and FedRep. The training efficiency of pFedSOP as compared to exiting methods is demonstrated through the Figures \ref{fig:1}, \ref{fig:2}, and \ref{fig:3}. We contract each of these figures, by plotting communication round wise average training loss and average test accuracy across all the available clients (i.e average over 20 $\%$ clients of total 100 clients in each communication round) for each method in each of the heterogeneous setting of each dataset. From these figures, it can be noticed that pFedSOP reduces training loss more quickly than existing methods with respect to communication rounds and achieves better test accuracy at various communication rounds. As, computation time per communication round of pFedSOP is comparatively lower, it may be claimed that pFedSOP can reduce losses and can achieved better accuracy in less time as compared to existing methods. As we use same initialization and same settings for all the methods while comparing, from the findings of these figures, we may claim that the efficiency of pFedSOP is better than existing methods in terms of accelerating the training of the personalized local models while targeting for better test accuracy from the personalized models.  
\begin{figure}[ht]
    \centering
    \begin{minipage}{0.25\textwidth}
        \centering
        \includegraphics[width=\linewidth]{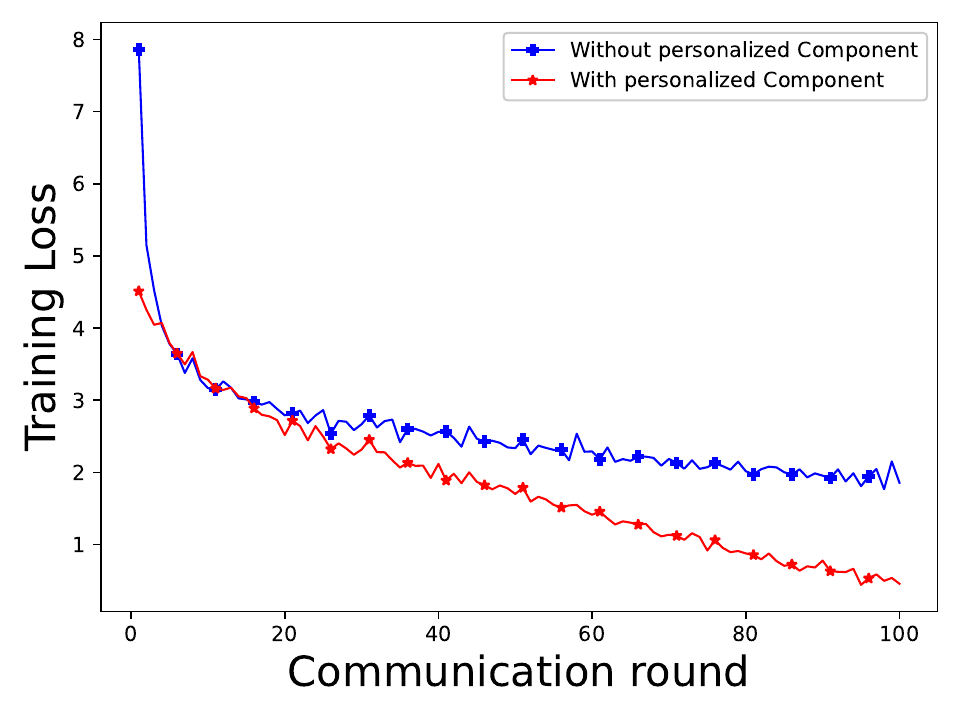}
    \end{minipage}
    \caption{Effect of personalization component (PC) on the performance of pFedSOP in term of communication round wise average training loss across all the available clients.}
    \label{fig:4}
\end{figure}

\subsubsection{Ablation Study}
This section demonstrates the effectiveness of the personalized component (PC) in pFedSOP, aimed at improving the PFL algorithm. To assess this, we conduct experiments comparing pFedSOP with and without the PC, evaluating their training efficiency in terms of training loss and test accuracy. These experiments are conducted on the CIFAR-100 dataset with a Dirichlet-based heterogeneous setting. The results are shown in Figure \ref{fig:4} and Table \ref{tab:2}. From these, it can be observed that pFedSOP with the PC improves FL training by more rapidly reducing training loss and achieving better test accuracy compared to pFedSOP without the PC.
\begin{table}[t]
\centering
\renewcommand{\arraystretch}{1.5}
\begin{tabular}{|c|c|c|}
\hline
Components  & without PC & with PC   \\
\hline
Accuracy  & 55.65 & 70.99    \\ 
\hline
\end{tabular}
\caption{Effect of personalized component (PC) in pFedSOP in term of best accuracy achieved. }
\label{tab:2}
\end{table}

\begin{table}[t]
\centering
\renewcommand{\arraystretch}{1.5}
\begin{minipage}{0.45\textwidth}
    \centering
    \begin{tabular}{|c|c|c|c|c|}
    \hline
    $\rho$  & 1 & 0.1 & 0.01 & 0.001   \\
    \hline
    Accuracy  & 70.99 & 71.36 & 71.12 & 71.24   \\ 
    \hline
    \hline
    \end{tabular}

\end{minipage}\hfill
\begin{minipage}{0.45\textwidth}
    \centering
    \begin{tabular}{|c|c|c|c|c|}
    \hline
    $\lambda$  & 5 & 2.5 & 1 & 0.5   \\
    \hline
    Accuracy  & 71.01 & 71.34 & 70.99 & 71.34   \\ 
    \hline
    \end{tabular}

\end{minipage}
\caption{Effect of $\rho$ and $\lambda$ on the performance of pFedSOP in term of best accuracy achieved.}
\label{tab:hyper}
\end{table}

\subsubsection{Parameters Sensitivity Analyses}
\begin{figure}[ht]
    \centering
    \begin{minipage}{0.23\textwidth} 
        \centering
        \includegraphics[width=\linewidth]{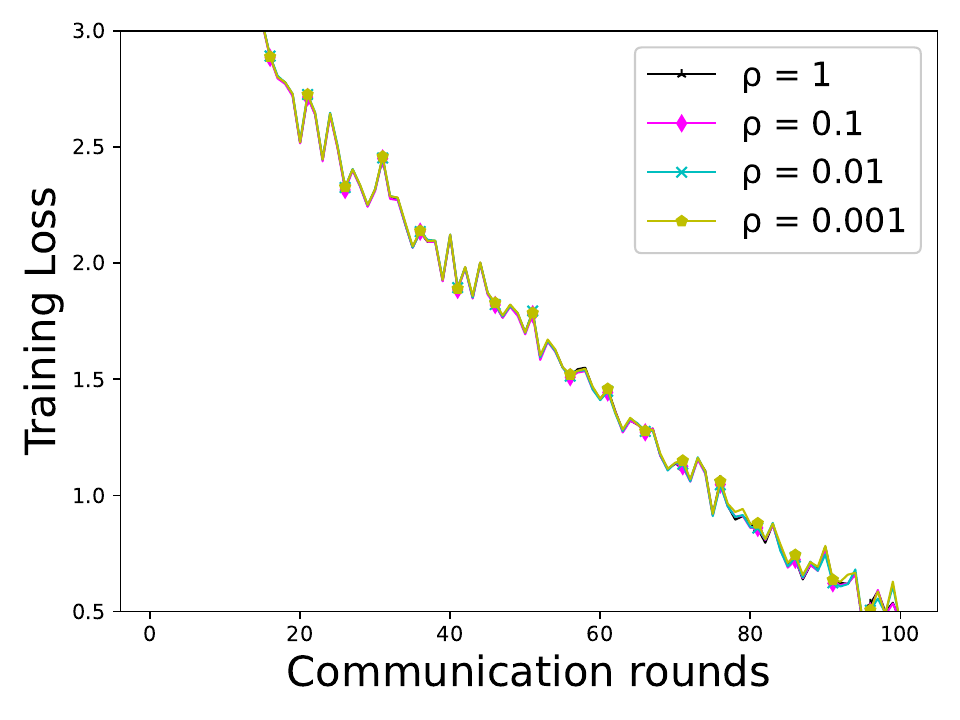}

    \end{minipage}\hfill
    \begin{minipage}{0.23\textwidth} 
        \centering
        \includegraphics[width=\linewidth]{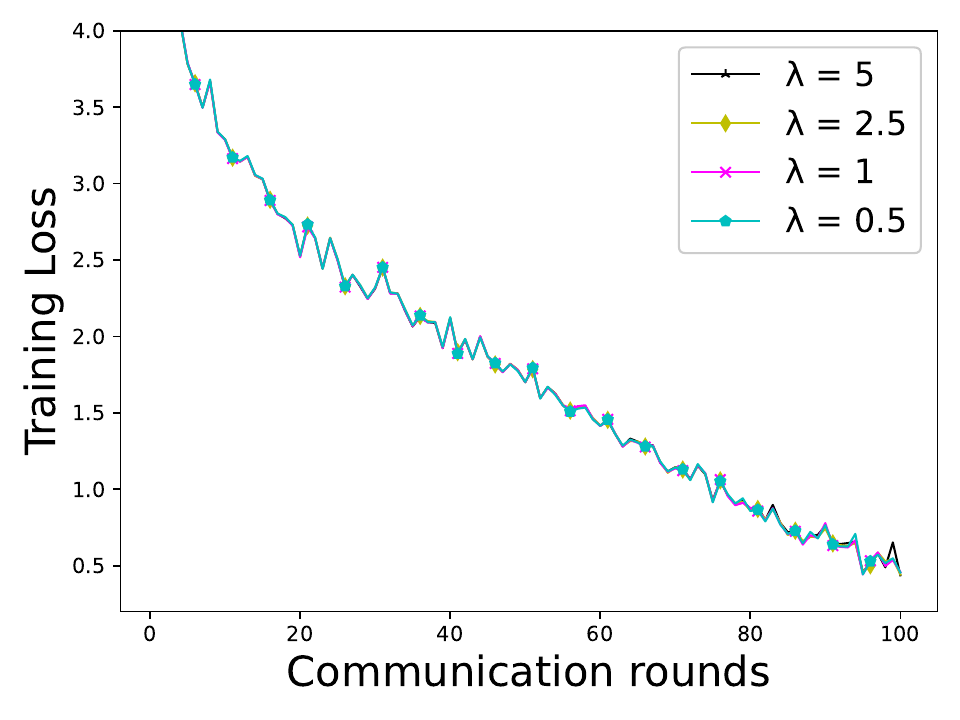}
    \end{minipage}
    \caption{Effect of $\rho$ and $\lambda$ on the performance of pFedSOp in term of communication round wise average training loss across all the available clients.}
    \label{fig:5}
\end{figure}
pFedSOP has two active hyperparameters. One is $\lambda$ of the Gompertz function, which is used to find the personalized update for each local client, and the other is the Hessian regularization parameter $\rho$, which is used for updating the personalized local models with second-order optimization. The effectiveness of these hyperparameters is analyzed through experiments on the CIFAR-100 dataset with the Dirichlet distribution-based FL setting, where we use a learning rate of $\eta = 0.01$. To analyze the sensitivity of $\rho$, we fix $\lambda = 1$ and compare the performance of pFedSOP for different values of $\rho \in \{1, 0.1, 0.01, 0.001\}$. To analyze $\lambda$, we fix $\rho = 1$ and compare the performance of pFedSOP for different values of $\lambda \in \{5, 2.5, 1, 0.5\}$. The results of these sensitivity analyses are shown in Figure \ref{fig:5} and Tables \ref{tab:hyper}. The findings from this figure and table suggest that, for better performance of pFedSOP, the value of $\rho$ should be chosen such that it is not far from the learning rate $\eta$. From this figure and table, we observe that there are no significant changes in the performance of pFedSOP when using $\lambda \in \{5, 2.5, 1, 0.5\}$. Any value of $\lambda \in \{5, 2.5, 1, 0.5\}$ can be recommended. 

\section{Conclusions and Future work}
\label{sc6}
Given the significant data heterogeneity across clients, a single global model in traditional federated learning often struggles with generalization, motivating the need for Personalized Federated Learning (PFL). However, existing PFL algorithms are hindered by slow training due to their reliance solely on first-order optimization. To address this, we propose pFedSOP, a method designed to enhance PFL efficiency by employing Fisher Information Matrix (FIM)-based second-order optimization, along with a Gompertz function-based personalized local gradient update for updating personalized local models. This approach accelerates the training of local models by efficiently utilizing second-order optimization along with relevant global information in each client. pFedSOP maintains competitive local computation costs compared to FedAvg, a common challenge in many existing PFL methods. Our experimental results demonstrate that pFedSOP significantly improves PFL training outcomes compared to existing federated learning approaches, while maintaining or reducing local computation costs. Looking ahead, we aim to extend our work to accelerate the training of personalized vertical federated learning.

\bibliographystyle{IEEEtran}
{\small
\bibliography{ref}}

\begin{IEEEbiography}[{\includegraphics[width=1in,height=1.25in,clip,keepaspectratio]{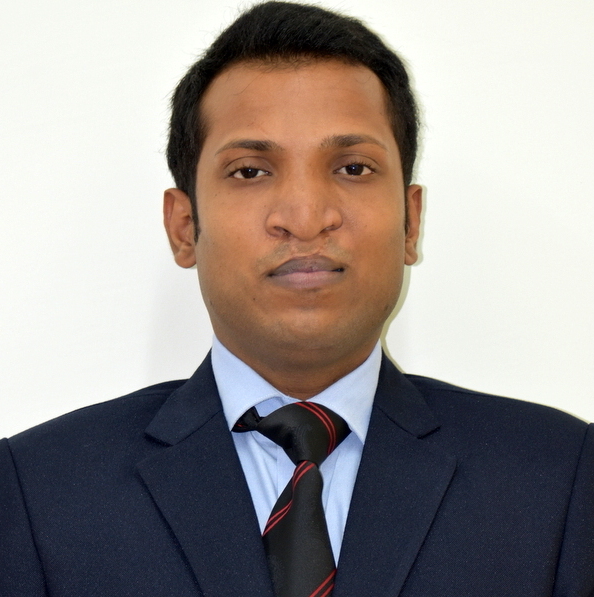}}]%
{Mrinmay Sen} is a joint research scholar in the Department of Artificial Intelligence, Indian Institute of Technology Hyderabad and the Department of Computing Technologies, Swinburne University of Technology, Hawthorn. He completed his M.Tech from Indian Institute of Technology Dhanbad and B.E. from Jadavpur University, Kolkata. His research interest broadly includes Federated optimization and its real life applications, machine learning and deep learning. 
\end{IEEEbiography}

\begin{IEEEbiography}[{\includegraphics[width=1in,height=1.25in,clip,keepaspectratio]{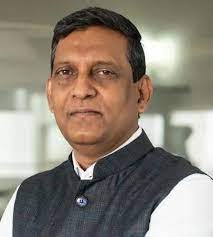}}]%
{Dr. C. Krishna Mohan} is currently a Professor in the Department of Computer Science and Engineering at the Indian Institute of Technology Hyderabad. His current research interests include video content analysis, computer vision, machine learning, deep learning. He has published more than 90 papers in various international journals and conference proceedings. He is a Senior Member of IEEE, Member of ACM, AAAI Member, Life Member of ISTE.
\end{IEEEbiography}
\end{document}